\documentclass[11pt]{article}
\usepackage{amsmath}
\usepackage{amsfonts}
\usepackage{amssymb}
\usepackage{graphicx}
\usepackage{fullpage}
\usepackage{longtable}
\usepackage{multicol}
\usepackage{hyperref}
\providecommand{\U}[1]{\protect\rule{.1in}{.1in}}
\newtheorem{theorem}{Theorem}

\newtheorem{conjecture}[theorem]{Conjecture}
\newtheorem{corollary}[theorem]{Corollary}

\newtheorem{definition}[theorem]{Definition}

\newtheorem{proposition}[theorem]{Proposition}

\newenvironment{proof}[1][Proof]{\noindent\textbf{#1.} }{\ \rule{0.5em}{0.5em}}
\hypersetup{
    colorlinks=true,
    linkcolor=blue,
    filecolor=magenta,      
    urlcolor=blue,
}
 
\begin{document}

\title{ An Upper Bound on the GKS Game via Max Bipartite Matching }
\author{DeVon Ingram\thanks{Email: dingram7@gatech.edu.}\\Georgia Institute of Technology}
\date{}
\maketitle

\begin{abstract}\noindent
The sensitivity conjecture is a longstanding conjecture concerning the relationship between the degree and sensitivity of a Boolean function. In 2015, a communication game was formulated by Justin Gilmer, Michal Kouck\'{y}, and Michael Saks to attempt to make progress on this conjecture. Andrew Drucker independently formulated this game. Shortly after the creation of the GKS game, Nisan Szegedy obtained a protocol for the game with a cost of $O(n^{.4732})$. We improve Szegedy's result to a cost of $O(n^{.4696})$ by providing a technique to identify whether a set of codewords can be used as a viable strategy in this game.
\end{abstract}

\section{Introduction\label{INTRO}}
In the field of query complexity, we are interested in the relationships between various complexity measures on Boolean functions such as deterministic query complexity, randomized query complexity, and polynomial degree. There are several surveys on the topic \cite{bw02}. We know most of these measures to be polynomial related \cite{ns94}. However, it is not known if sensitivity is polynomial related to these complexity measures. The original formulation of the sensitivity conjecture asks whether there is a constant $d$ such that $deg(f)$ = $O$($s(f)^d$). There have been many approaches to settling the sensitivity conjecture \cite{gs004}. Many of the approaches shifts the attention from degree and towards other complexity measures \cite{kk04,gstw16,apv15}. \par
In this paper, we are particularly interested in the communication game formulation of the sensitivity conjecture \cite{sz15,gks15,d17}. One of the main results of \cite{gks15} is that the cost of the game played on $deg(f)$ variables is at most $s(f)$, where $f$ is an arbitrary Boolean function. A polynomial lower bound on the cost of the game implies that the sensitivity conjecture is true. The GKS communication game is played with two players cooperating against a third party, Eve. Eve first sends a permutation of the set of positive integers from 1 to $n$, denoted by [$n$], to Alice. This set represents indices of a binary string which Alice and Eve will both fill in. Once Alice receives the permutation, one at a time, she places either a 0 or 1 at each specified location with no knowledge of the indices she will edit next. Once Alice places a bit at the $n-1^{th}$ index, Eve fills in the last bit and sends the string to Bob, who tries to determine which bit Eve placed. Bob returns a subset $S$ of [$n$] which must contain the index of the bit Eve placed. The cost of the game is then given by $C$($n$) = $|S|$. In 2015, Gilmer, Kouck\'{y}, and Saks \cite{gks15} provided a ${\left\lfloor\sqrt{n}\right\rfloor}$ lower bound for monotone protocols. Later that year, Mario Szegedy \cite{sz15} proved that there is an $O(n^{.4732})$ upper bound on the cost of the game (non-monotone). We improve this upper bound to $O(n^{.4696})$.

\section{Background\label{BACKGROUND}}
We will begin by defining some relevant query complexity measures.

\begin{definition}
Let $f:\{0,1\}^n\rightarrow\{0,1\}$ be a boolean function. The decision tree complexity of $f$ on an input $x$ is the minimum number of input bits that a deterministic classical algorithm needs to query to determine the value of $f$ on the input $x$. The decision tree complexity is the maximum of this quantity over all $x\in\{0,1\}^n$.
\end{definition}

\begin{definition}
Let $f:\{0,1\}^n\rightarrow\{0,1\}$ be a boolean function. The sensitivity of $f$ on input $x$ is the number of indices $i$ such that $f(x) = 1 - f(x')$ where $x'$ represents the input $x$ with the bit $x_i$ flipped. The sensitivity of the function, denoted $s(f)$, is the maximum sensitivity over all inputs.
\end{definition}

\begin{definition}
Let $f:\{0,1\}^n\rightarrow\{0,1\}$ be a boolean function. The degree of $f$ is the number of different variables in the largest monomial in its multilinear polynomial representation.
\end{definition}

\noindent A natural question to ask when presented with multiple different complexity measures is "How are these measures related?". We are familiar with relations between many natural measures already \cite{ns94}. We ultimately want to determine the relationship between sensitivity and the other common complexity measures. One approach to this is studying the relationship between sensitivity and degree.

\begin{conjecture}(Sensitivity Conjecture)
There exists a $d\in\mathbb{R}$ such that $deg(f)\leq s(f)^d$ for all boolean functions $f$.
\end{conjecture}

\noindent The conjecture above is one formulation of the sensitivity conjecture. The GKS game provides us an approach to settling the conjecture through the following proposition:

\begin{proposition}
\cite{gks15} Let $C(n)$ denote the cost of the GKS game played on $n$ variables. Then $s(f)\leq C(deg(f))$.
\end{proposition}

\noindent From the proposition above, it is clear that a polynomial lower bound on the cost of the GKS game implies that degree and sensitivity are polynomially related. 

\begin{definition}
Let the GKS game be played with a binary string of length $n$ and let $k$ be the size of the subset returned by Bob as a result of Alice's strategy in the worst case. We say Alice's strategy is a ($k$,$n$) strategy. 
\end{definition}

 \noindent Szegedy provided a (5,30) strategy for the communication game by making use of 6 distinct binary strings that could be distinguished even after a single bit flip occurs in any of the digits \cite{sz15}. Hamming codewords also satisfy this property, so this gives us a way to generalize this technique for $n=2^k$ where $k\in\mathbb{Z^+}$, the length of the Hamming code.

\begin{definition}
The Hamming code is a linear error-correcting code with a parity check matrix $H$ given by $H = \left(A | I_{n-k}\right)$ where $I_{n-k}$ is the ($n$-$k$)-identity matrix and $A$ is all the nonzero n-tuples that do not appear in $I_{n-k}$.
\end{definition}

\begin{proposition}\label{szprop}
If there exists a ($k$, $n$) strategy and a there exists a ($k'$, $n'$) strategy, then there exists a ($kk'$, $nn'$) strategy.
\end{proposition}

\noindent The proof of this proposition was provided by Szegedy \cite{sz15}. Because of the proposition above, we can make crude approximations to the cost of the game for large values of $n$ by repeatedly composing smaller games, for which strategies are known, together. This allows us to obtain an upper bound on the cost of the game after determining the existence of a $(k,n)$ strategy.

\section{Results\label{RESULTS}}
\begin{theorem}
There exists an (11, 165) strategy for the GKS game.
\end{theorem}
\begin{proof}
First, we note that the Hamming code of length 15 has dimension 11 and distance 3. This implies that there exist a set of 2048 codewords of length 15 such that single bit flips that occur on these codewords can be distinguished. We used a computer search to show that there exists a bipartite matching of size $\binom{15}{4}$ between 4 digit combinations of indices all filled with ones and the set of codewords (see Appendix A). Given this matching, Alice proceeds with the following protocol: \\
1) Break up the input into 11 blocks of size 15. \\
2) For the first 4 indices received that correspond to a single block, place ones. \\
3) For the 5$^{th}$ index received and so on, fill in bits corresponding to the binary codeword to which the 4 index combination is matched. \\
4) For the last bit in a single block, Alice sets a bit so that the block is not a codeword. \\
The maximum size of the subset Bob must return is 11. If Eve places her bit so that the last block is not a codeword, then Bob knows the last bit flipped in every block and returns the eleven indices associated with those bits. If Eve set her bit so that the last block is a codeword, then that block will be the only block that is a codeword. Then, Bob can use the matching to figure out which four indices were set first so that he can return the other eleven.
\end{proof}

\noindent To obtain an upper bound on the cost of this game, we compose games of size 165 together while using our (11,165) strategy for each game to obtain the following:

\begin{corollary}
There exists an ($11^k$, $165^k$) strategy for the GKS game.
\end{corollary}

\noindent The corollary above is an application of Proposition~\ref{szprop}. With this result, we can place an upper bound on the cost ($C(n) = O(n^{.4696})$) of the GKS game.

\begin{corollary}
$C(n) < 11 \cdot n^{log_{165}11}$
\end{corollary}

\section{Acknowledgments\label{ACKNOWLEDGEMENTS}}
I would like to thank Lance Fortnow for helpful discussions.
\pagebreak
\bibliographystyle{plain}
\bibliography{sources}
\pagebreak
\appendix
\section{Max Bipartite Matching}
Below is the explicit matching referred to earlier in the paper. Each subset of four indices is matched to a hamming codeword. Each string in the table satisfies two properties:
\begin{enumerate}
    \item For every subset of size 4 of $\{1,\cdots, 15\}$ and for every $i$ in a subset, the binary string matched to it has a 1 at bit $x_i$.
    \item Let $S,T$ be two distinct subsets of size 4 of $\{1,\cdots, 15\}$. Then, their matched strings differ in at least three bits. 
\end{enumerate}
To see the code that produced this matching, see \href{https://github.com/GKSMaxBipartiteMatching/GKS-Game-Strategy}{the GitHub repository} \cite{dg17}.
\begin{longtable}{|p{2.2cm}||p{3cm}||p{2.2cm}||p{3cm}||p{2.2cm}||p{3cm}|}
\hline
\multicolumn{6}{|c|}{Bipartite Matching} \\
\hline
Subset & Matched String & Subset & Matched String & Subset & Matched String \\
\hline
\{1,2,3,4\}&111111001000010&\{2,4,7,14\}&010100110100110&\{4,5,12,14\}&001110000001010\\
\{1,2,3,5\}&111011101000100&\{2,4,7,15\}&010101110000001&\{4,5,12,15\}&010110000001001\\
\{1,2,3,6\}&111101000111001&\{2,4,8,9\}&011110011110000&\{4,5,13,14\}&010110000000110\\
\{1,2,3,7\}&111011110110100&\{2,4,8,10\}&010101110110010&\{4,5,13,15\}&110110000000101\\
\{1,2,3,8\}&111110111011000&\{2,4,8,11\}&110110010011100&\{4,5,14,15\}&000110000000011\\
\{1,2,3,9\}&111001101011001&\{2,4,8,12\}&011110110001110&\{4,6,7,8\}&000101111100010\\
\{1,2,3,10\}&111111110101010&\{2,4,8,13\}&010111011011110&\{4,6,7,9\}&000101101001000\\
\{1,2,3,11\}&111111000011000&\{2,4,8,14\}&010100010000010&\{4,6,7,10\}&000101110111000\\
\{1,2,3,12\}&111101111011110&\{2,4,8,15\}&110100010000001&\{4,6,7,11\}&000101111011110\\
\{1,2,3,13\}&111111100111100&\{2,4,9,10\}&011111001110010&\{4,6,7,12\}&000101100101110\\
\{1,2,3,14\}&111110101001110&\{2,4,9,11\}&010111101010000&\{4,6,7,13\}&000101110000100\\
\{1,2,3,15\}&111111010000001&\{2,4,9,12\}&010110101111000&\{4,6,7,14\}&000101100010010\\
\{1,2,4,5\}&110111011101110&\{2,4,9,13\}&010101101111110&\{4,6,7,15\}&100101111100001\\
\{1,2,4,6\}&111101100010010&\{2,4,9,14\}&110100001000010&\{4,6,8,9\}&110101011000000\\
\{1,2,4,7\}&110100100000000&\{2,4,9,15\}&010100001000001&\{4,6,8,10\}&000101010100000\\
\{1,2,4,8\}&111111110010110&\{2,4,10,11\}&010100010111110&\{4,6,8,11\}&100101010010000\\
\{1,2,4,9\}&111111111110000&\{2,4,10,12\}&010100000101000&\{4,6,8,12\}&100101110001000\\
\{1,2,4,10\}&110101101110010&\{2,4,10,13\}&110100000100100&\{4,6,8,13\}&010111010000100\\
\{1,2,4,11\}&111111110011001&\{2,4,10,14\}&010111011100010&\{4,6,8,14\}&100111010000010\\
\{1,2,4,12\}&110111111001010&\{2,4,10,15\}&011100000100001&\{4,6,8,15\}&000111010000001\\
\{1,2,4,13\}&111111001111110&\{2,4,11,12\}&110100000011000&\{4,6,9,10\}&100101001100000\\
\{1,2,4,14\}&111110110000010&\{2,4,11,13\}&010100000010100&\{4,6,9,11\}&000101001010000\\
\{1,2,4,15\}&110110000111001&\{2,4,11,14\}&111100011010010&\{4,6,9,12\}&101111001001000\\
\{1,2,5,6\}&110011000000000&\{2,4,11,15\}&010100101011001&\{4,6,9,13\}&001111001000100\\
\{1,2,5,7\}&111111100000000&\{2,4,12,13\}&010110010101100&\{4,6,9,14\}&000111001000010\\
\{1,2,5,8\}&111111111001100&\{2,4,12,14\}&011101000111010&\{4,6,9,15\}&100111001000001\\
\{1,2,5,9\}&111111101011010&\{2,4,12,15\}&010111111001001&\{4,6,10,11\}&001101000110000\\
\{1,2,5,10\}&110010111111100&\{2,4,13,14\}&011100111000110&\{4,6,10,12\}&000101001101100\\
\{1,2,5,11\}&111010010111000&\{2,4,13,15\}&010111100000101&\{4,6,10,13\}&000101101110100\\
\{1,2,5,12\}&110111101011100&\{2,4,14,15\}&010100100000011&\{4,6,10,14\}&001111000100010\\
\{1,2,5,13\}&110111111110110&\{2,5,6,7\}&011011111100010&\{4,6,10,15\}&000101100100001\\
\{1,2,5,14\}&111011011110110&\{2,5,6,8\}&011011010011100&\{4,6,11,12\}&100101000111010\\
\{1,2,5,15\}&110111011100001&\{2,5,6,9\}&010011111100100&\{4,6,11,13\}&011111000010100\\
\{1,2,6,7\}&110101111100100&\{2,5,6,10\}&111011001100000&\{4,6,11,14\}&010111000010010\\
\{1,2,6,8\}&111011111101110&\{2,5,6,11\}&010011110111110&\{4,6,11,15\}&001101001011001\\
\{1,2,6,9\}&110101001011001&\{2,5,6,12\}&010011111011000&\{4,6,12,13\}&110101000001100\\
\{1,2,6,10\}&110101000110000&\{2,5,6,13\}&110111001000100&\{4,6,12,14\}&000101000001010\\
\{1,2,6,11\}&111111011010100&\{2,5,6,14\}&111011000000110&\{4,6,12,15\}&100101000001001\\
\{1,2,6,12\}&110001111111010&\{2,5,6,15\}&110011110000001&\{4,6,13,14\}&100101000000110\\
\{1,2,6,13\}&111001111111100&\{2,5,7,8\}&010010111110000&\{4,6,13,15\}&000101000000101\\
\{1,2,6,14\}&111111010001110&\{2,5,7,9\}&011011101001000&\{4,6,14,15\}&100101010100011\\
\{1,2,6,15\}&111111001110001&\{2,5,7,10\}&011010110101100&\{4,7,8,9\}&010100111000000\\
\{1,2,7,8\}&111101110000100&\{2,5,7,11\}&010011101110010&\{4,7,8,10\}&100100110100000\\
\{1,2,7,9\}&111011111010010&\{2,5,7,12\}&011010111001010&\{4,7,8,11\}&000100110010000\\
\{1,2,7,10\}&110111101100000&\{2,5,7,13\}&010010111001100&\{4,7,8,12\}&010110110001000\\
\{1,2,7,11\}&111000101111110&\{2,5,7,14\}&010010101011010&\{4,7,8,13\}&000100110101100\\
\{1,2,7,12\}&111010111111010&\{2,5,7,15\}&010011110110001&\{4,7,8,14\}&000110110000010\\
\{1,2,7,13\}&111110111100100&\{2,5,8,9\}&010011011111100&\{4,7,8,15\}&011110110000001\\
\{1,2,7,14\}&110110111011110&\{2,5,8,10\}&010111010111000&\{4,7,9,10\}&000100101100000\\
\{1,2,7,15\}&111111101101001&\{2,5,8,11\}&010011010011010&\{4,7,9,11\}&100100101010000\\
\{1,2,8,9\}&111101011111010&\{2,5,8,12\}&010010110101010&\{4,7,9,12\}&001110101001000\\
\{1,2,8,10\}&110101011111100&\{2,5,8,13\}&011011110000100&\{4,7,9,13\}&010110101000100\\
\{1,2,8,11\}&110000010010000&\{2,5,8,14\}&110010010000010&\{4,7,9,14\}&011110101000010\\
\{1,2,8,12\}&111100011101110&\{2,5,8,15\}&010010010000001&\{4,7,9,15\}&000110101000001\\
\{1,2,8,13\}&111110011111100&\{2,5,9,10\}&111010011100010&\{4,7,10,11\}&010100100110000\\
\{1,2,8,14\}&111101111100010&\{2,5,9,11\}&110110001010000&\{4,7,10,12\}&001100110101010\\
\{1,2,8,15\}&111101011001001&\{2,5,9,12\}&011011001101100&\{4,7,10,13\}&100110100100100\\
\{1,2,9,10\}&110000001100000&\{2,5,9,13\}&010111001110100&\{4,7,10,14\}&101110100100010\\
\{1,2,9,11\}&111110001011001&\{2,5,9,14\}&010010001000010&\{4,7,10,15\}&110110100100001\\
\{1,2,9,12\}&111011101111000&\{2,5,9,15\}&110010001000001&\{4,7,11,12\}&010101100011000\\
\{1,2,9,13\}&111101101110100&\{2,5,10,11\}&010010001111110&\{4,7,11,13\}&000110100010100\\
\{1,2,9,14\}&111001101101010&\{2,5,10,12\}&110010000101000&\{4,7,11,14\}&001110100010010\\
\{1,2,9,15\}&110010101011001&\{2,5,10,13\}&010010000100100&\{4,7,11,15\}&110100110011001\\
\{1,2,10,11\}&111110110111110&\{2,5,10,14\}&010011010100110&\{4,7,12,13\}&101100100001100\\
\{1,2,10,12\}&111111011101000&\{2,5,10,15\}&011111010110001&\{4,7,12,14\}&100100100001010\\
\{1,2,10,13\}&111101000110110&\{2,5,11,12\}&010010000011000&\{4,7,12,15\}&000100100001001\\
\{1,2,10,14\}&111110101110010&\{2,5,11,13\}&110010000010100&\{4,7,13,14\}&000100100000110\\
\{1,2,10,15\}&111100011100001&\{2,5,11,14\}&010110100011110&\{4,7,13,15\}&100100100000101\\
\{1,2,11,12\}&110111100111010&\{2,5,11,15\}&011110101110001&\{4,7,14,15\}&000111111000011\\
\{1,2,11,13\}&111011100011110&\{2,5,12,13\}&010110001011100&\{4,8,9,10\}&100100011100010\\
\{1,2,11,14\}&111100111110110&\{2,5,12,14\}&110110000001010&\{4,8,9,11\}&000111011010100\\
\{1,2,11,15\}&110011101110001&\{2,5,12,15\}&111011000001001&\{4,8,9,12\}&100101011001010\\
\{1,2,12,13\}&111101100101110&\{2,5,13,14\}&110111100000110&\{4,8,9,13\}&000100111110110\\
\{1,2,12,14\}&111011011001010&\{2,5,13,15\}&011010111000101&\{4,8,9,14\}&000100011010010\\
\{1,2,12,15\}&110000000001001&\{2,5,14,15\}&111111111000011&\{4,8,9,15\}&100100011010001\\
\{1,2,13,14\}&110000000000110&\{2,6,7,8\}&010001111001010&\{4,8,10,11\}&000101011111010\\
\{1,2,13,15\}&110100110100101&\{2,6,7,9\}&110001101010000&\{4,8,10,12\}&000100011101110\\
\{1,2,14,15\}&111001100000011&\{2,6,7,10\}&110001110100000&\{4,8,10,13\}&000100010110100\\
\{1,3,4,5\}&101111111111010&\{2,6,7,11\}&010001101011100&\{4,8,10,14\}&100110010101010\\
\{1,3,4,6\}&101101000000000&\{2,6,7,12\}&010001110101100&\{4,8,10,15\}&010100010110001\\
\{1,3,4,7\}&101111110100000&\{2,6,7,13\}&011101101000100&\{4,8,11,12\}&001110010011100\\
\{1,3,4,8\}&111110011000000&\{2,6,7,14\}&010001111110110&\{4,8,11,13\}&000101010011100\\
\{1,3,4,9\}&101110101000100&\{2,6,7,15\}&011101100010001&\{4,8,11,14\}&110100010110010\\
\{1,3,4,10\}&101111101101100&\{2,6,8,9\}&111001011011000&\{4,8,11,15\}&100110010011001\\
\{1,3,4,11\}&101101111010100&\{2,6,8,10\}&010001010110100&\{4,8,12,13\}&100110011001100\\
\{1,3,4,12\}&101111110011100&\{2,6,8,11\}&011101010010000&\{4,8,12,14\}&000100111001010\\
\{1,3,4,13\}&101111111000110&\{2,6,8,12\}&010001010001000&\{4,8,12,15\}&001100110011001\\
\{1,3,4,14\}&101111100110110&\{2,6,8,13\}&110001010000100&\{4,8,13,14\}&110110011000110\\
\{1,3,4,15\}&101110111100001&\{2,6,8,14\}&110101110000010&\{4,8,13,15\}&010110011000101\\
\{1,3,5,6\}&101011011000000&\{2,6,8,15\}&111001010110001&\{4,8,14,15\}&101101011000011\\
\{1,3,5,7\}&101010100000000&\{2,6,9,10\}&011101001100000&\{4,9,10,11\}&000110101110010\\
\{1,3,5,8\}&101011011111100&\{2,6,9,11\}&110001001110100&\{4,9,10,12\}&000110001101010\\
\{1,3,5,9\}&101110101111000&\{2,6,9,12\}&110001001001000&\{4,9,10,13\}&100100001110100\\
\{1,3,5,10\}&101011111100100&\{2,6,9,13\}&010001001000100&\{4,9,10,14\}&100110001100110\\
\{1,3,5,11\}&101111010111000&\{2,6,9,14\}&011001001111110&\{4,9,10,15\}&000100011100001\\
\{1,3,5,12\}&111010100001010&\{2,6,9,15\}&010101101110001&\{4,9,11,12\}&000111101011010\\
\{1,3,5,13\}&101011110111110&\{2,6,10,11\}&110001010111000&\{4,9,11,13\}&001100011010100\\
\{1,3,5,14\}&101111000101110&\{2,6,10,12\}&010101010101010&\{4,9,11,14\}&001101001010110\\
\{1,3,5,15\}&111110110110001&\{2,6,10,13\}&010101100100100&\{4,9,11,15\}&000101111010001\\
\{1,3,6,7\}&111011110001000&\{2,6,10,14\}&010001000100010&\{4,9,12,13\}&100101001011100\\
\{1,3,6,8\}&101011111011000&\{2,6,10,15\}&110001000100001&\{4,9,12,14\}&100110001011010\\
\{1,3,6,9\}&101101011110000&\{2,6,11,12\}&010101000111100&\{4,9,12,15\}&000101011001001\\
\{1,3,6,10\}&101101101111110&\{2,6,11,13\}&111001000010100&\{4,9,13,14\}&001110011000110\\
\{1,3,6,11\}&101111011011110&\{2,6,11,14\}&110001000010010&\{4,9,13,15\}&001111111000101\\
\{1,3,6,12\}&111101101001000&\{2,6,11,15\}&010001000010001&\{4,9,14,15\}&100110011000011\\
\{1,3,6,13\}&101111001110100&\{2,6,12,13\}&111001100001100&\{4,10,11,12\}&100110000111100\\
\{1,3,6,14\}&111111010110010&\{2,6,12,14\}&011011000001010&\{4,10,11,13\}&001110000110110\\
\{1,3,6,15\}&101011010101001&\{2,6,12,15\}&010101010011001&\{4,10,11,14\}&010110000111010\\
\{1,3,7,8\}&111001111000000&\{2,6,13,14\}&010111000101110&\{4,10,11,15\}&000110110110001\\
\{1,3,7,9\}&101110111101110&\{2,6,13,15\}&110111111000101&\{4,10,12,13\}&100101010101100\\
\{1,3,7,10\}&101100111111100&\{2,6,14,15\}&010001110100011&\{4,10,12,14\}&011100000101110\\
\{1,3,7,11\}&101111101010000&\{2,7,8,9\}&110000111010010&\{4,10,12,15\}&010101001101001\\
\{1,3,7,12\}&101001101011100&\{2,7,8,10\}&011000111100100&\{4,10,13,14\}&000101000110110\\
\{1,3,7,13\}&101000111011110&\{2,7,8,11\}&110000110110100&\{4,10,13,15\}&010101010100101\\
\{1,3,7,14\}&101110111010010&\{2,7,8,12\}&110000110001000&\{4,10,14,15\}&101110010100011\\
\{1,3,7,15\}&111101111010001&\{2,7,8,13\}&010000110000100&\{4,11,12,13\}&100100110011100\\
\{1,3,8,9\}&101111011100010&\{2,7,8,14\}&010000111011110&\{4,11,12,14\}&000100000011110\\
\{1,3,8,10\}&101000010100000&\{2,7,8,15\}&011001110011001&\{4,11,12,15\}&000110001011001\\
\{1,3,8,11\}&101100010111110&\{2,7,9,10\}&011100101101100&\{4,11,13,14\}&000110001010110\\
\{1,3,8,12\}&111100110101100&\{2,7,9,11\}&111010101010000&\{4,11,13,15\}&001101010010101\\
\{1,3,8,13\}&101110011110110&\{2,7,9,12\}&010000101001000&\{4,11,14,15\}&110110010010011\\
\{1,3,8,14\}&111110010100110&\{2,7,9,13\}&110000101000100&\{4,12,13,14\}&000111010001110\\
\{1,3,8,15\}&101100110101001&\{2,7,9,14\}&010100101101010&\{4,12,13,15\}&101100010001101\\
\{1,3,9,10\}&101010111110000&\{2,7,9,15\}&010110111100001&\{4,12,14,15\}&100100010001011\\
\{1,3,9,11\}&101000001010000&\{2,7,10,11\}&010010100111100&\{4,13,14,15\}&111100010000111\\
\{1,3,9,12\}&111100101011100&\{2,7,10,12\}&110101100101000&\{5,6,7,8\}&000011111010010\\
\{1,3,9,13\}&101010101100110&\{2,7,10,13\}&111000100100100&\{5,6,7,9\}&000011111101110\\
\{1,3,9,14\}&101011101110010&\{2,7,10,14\}&110000100100010&\{5,6,7,10\}&111011100100010\\
\{1,3,9,15\}&101111011010001&\{2,7,10,15\}&010000100100001&\{5,6,7,11\}&000011110110100\\
\{1,3,10,11\}&111011000111010&\{2,7,11,12\}&011110100011000&\{5,6,7,12\}&100011101001000\\
\{1,3,10,12\}&111101001101100&\{2,7,11,13\}&111000111010100&\{5,6,7,13\}&100011110000100\\
\{1,3,10,13\}&101001111110110&\{2,7,11,14\}&010000100010010&\{5,6,7,14\}&000011100100010\\
\{1,3,10,14\}&111001001110010&\{2,7,11,15\}&110000100010001&\{5,6,7,15\}&101011101000001\\
\{1,3,10,15\}&101000000111001&\{2,7,12,13\}&010100100001100&\{5,6,8,9\}&010011011000000\\
\{1,3,11,12\}&111101110111000&\{2,7,12,14\}&011100100001010&\{5,6,8,10\}&100011010100000\\
\{1,3,11,13\}&101110110110100&\{2,7,12,15\}&111100100001001&\{5,6,8,11\}&000011010010000\\
\{1,3,11,14\}&101110100011110&\{2,7,13,14\}&111100100000110&\{5,6,8,12\}&000011110001000\\
\{1,3,11,15\}&111010011010001&\{2,7,13,15\}&011100100000101&\{5,6,8,13\}&100011010011100\\
\{1,3,12,13\}&111010101101100&\{2,7,14,15\}&011111100000011&\{5,6,8,14\}&100011011000110\\
\{1,3,12,14\}&101000000001010&\{2,8,9,10\}&011000011111100&\{5,6,8,15\}&100011011001001\\
\{1,3,12,15\}&101101001101001&\{2,8,9,11\}&010000011111010&\{5,6,9,10\}&000011001100000\\
\{1,3,13,14\}&111001010111110&\{2,8,9,12\}&011101011001010&\{5,6,9,11\}&100011001010000\\
\{1,3,13,15\}&101000000000101&\{2,8,9,13\}&010000011000110&\{5,6,9,12\}&000011101111000\\
\{1,3,14,15\}&111011010100011&\{2,8,9,14\}&011100011100010&\{5,6,9,13\}&000011101000100\\
\{1,4,5,6\}&110111110101100&\{2,8,9,15\}&110000111100001&\{5,6,9,14\}&010011001010110\\
\{1,4,5,7\}&110111110010000&\{2,8,10,11\}&011010010110100&\{5,6,9,15\}&010011101000001\\
\{1,4,5,8\}&100111111000000&\{2,8,10,12\}&011101010101100&\{5,6,10,11\}&010011000110000\\
\{1,4,5,9\}&100111011100100&\{2,8,10,13\}&010001011101110&\{5,6,10,12\}&100111000101000\\
\{1,4,5,10\}&100111101101010&\{2,8,10,14\}&111000010101010&\{5,6,10,13\}&110011100100100\\
\{1,4,5,11\}&100111010111110&\{2,8,10,15\}&010011010101001&\{5,6,10,14\}&110111000100010\\
\{1,4,5,12\}&100111100001100&\{2,8,11,12\}&011000010011010&\{5,6,10,15\}&100011100100001\\
\{1,4,5,13\}&100111101010110&\{2,8,11,13\}&111000010010110&\{5,6,11,12\}&001011100011000\\
\{1,4,5,14\}&110111011010010&\{2,8,11,14\}&011010011010010&\{5,6,11,13\}&010011100010100\\
\{1,4,5,15\}&101110000001001&\{2,8,11,15\}&011100011010001&\{5,6,11,14\}&011011100010010\\
\{1,4,6,7\}&101101111101000&\{2,8,12,13\}&111000011001100&\{5,6,11,15\}&000011100010001\\
\{1,4,6,8\}&110101110111110&\{2,8,12,14\}&010110011001010&\{5,6,12,13\}&010011000001100\\
\{1,4,6,9\}&100111111111100&\{2,8,12,15\}&010010110011001&\{5,6,12,14\}&100011000001010\\
\{1,4,6,10\}&100101111101110&\{2,8,13,14\}&011110010010110&\{5,6,12,15\}&000011000001001\\
\{1,4,6,11\}&100101111010010&\{2,8,13,15\}&111000010100101&\{5,6,13,14\}&000011000000110\\
\{1,4,6,12\}&110101111011000&\{2,8,14,15\}&011101010100011&\{5,6,13,15\}&100011000000101\\
\{1,4,6,13\}&100101100011110&\{2,9,10,11\}&111010001110100&\{5,6,14,15\}&010011000000011\\
\{1,4,6,14\}&100111110011010&\{2,9,10,12\}&111100001111000&\{5,7,8,9\}&110010111000000\\
\{1,4,6,15\}&111101100100001&\{2,9,10,13\}&011100001110100&\{5,7,8,10\}&000010110100000\\
\{1,4,7,8\}&111100111001010&\{2,9,10,14\}&010011001101010&\{5,7,8,11\}&100010110010000\\
\{1,4,7,9\}&101101101000010&\{2,9,10,15\}&010010001110001&\{5,7,8,12\}&000010111111010\\
\{1,4,7,10\}&110100111110000&\{2,9,11,12\}&010101001011010&\{5,7,8,13\}&110110110000100\\
\{1,4,7,11\}&100110101111110&\{2,9,11,13\}&010100101010110&\{5,7,8,14\}&010011110000010\\
\{1,4,7,12\}&100110111101000&\{2,9,11,14\}&011000101110010&\{5,7,8,15\}&001011110000001\\
\{1,4,7,13\}&100110110001110&\{2,9,11,15\}&011001001110001&\{5,7,9,10\}&100010101100000\\
\{1,4,7,14\}&100100111111010&\{2,9,12,13\}&010101011001100&\{5,7,9,11\}&000010101010000\\
\{1,4,7,15\}&100111101011001&\{2,9,12,14\}&011000001101010&\{5,7,9,12\}&100010111001010\\
\{1,4,8,9\}&110110011111010&\{2,9,12,15\}&010011001011001&\{5,7,9,13\}&000010101101100\\
\{1,4,8,10\}&110100110101010&\{2,9,13,14\}&110101001010110&\{5,7,9,14\}&010010101100110\\
\{1,4,8,11\}&111110010011010&\{2,9,13,15\}&010001111000101&\{5,7,9,15\}&010010101101001\\
\{1,4,8,12\}&110110110111000&\{2,9,14,15\}&010010111000011&\{5,7,10,11\}&110010100110000\\
\{1,4,8,13\}&100100010000100&\{2,10,11,12\}&011110000111100&\{5,7,10,12\}&000110100101000\\
\{1,4,8,14\}&100111110100110&\{2,10,11,13\}&011011000110110&\{5,7,10,13\}&011110100100100\\
\{1,4,8,15\}&110110011001001&\{2,10,11,14\}&011001010110010&\{5,7,10,14\}&010110100100010\\
\{1,4,9,10\}&110110101110100&\{2,10,11,15\}&010000000111001&\{5,7,10,15\}&000011111100001\\
\{1,4,9,11\}&100110111010100&\{2,10,12,13\}&110000010101100&\{5,7,11,12\}&100110100011000\\
\{1,4,9,12\}&100100001001000&\{2,10,12,14\}&011110010101010&\{5,7,11,13\}&111110100010100\\
\{1,4,9,13\}&110100101100110&\{2,10,12,15\}&011110001101001&\{5,7,11,14\}&110110100010010\\
\{1,4,9,14\}&100101011110110&\{2,10,13,14\}&111000001100110&\{5,7,11,15\}&010110100010001\\
\{1,4,9,15\}&111110101000001&\{2,10,13,15\}&010010110100101&\{5,7,12,13\}&110010100001100\\
\{1,4,10,11\}&101101110110010&\{2,10,14,15\}&011010110100011&\{5,7,12,14\}&000010100001010\\
\{1,4,10,12\}&100101101111000&\{2,11,12,13\}&011101001011100&\{5,7,12,15\}&100010100001001\\
\{1,4,10,13\}&111111000100100&\{2,11,12,14\}&010100110011010&\{5,7,13,14\}&100010100000110\\
\{1,4,10,14\}&110110111100010&\{2,11,12,15\}&011000001011001&\{5,7,13,15\}&000010100000101\\
\{1,4,10,15\}&100100000100001&\{2,11,13,14\}&011000001010110&\{5,7,14,15\}&110010100000011\\
\{1,4,11,12\}&100111011011000&\{2,11,13,15\}&010011010010101&\{5,8,9,10\}&001010011100100\\
\{1,4,11,13\}&100101110110100&\{2,11,14,15\}&111101010010011&\{5,8,9,11\}&101010011010100\\
\{1,4,11,14\}&100100000010010&\{2,12,13,14\}&010010010001110&\{5,8,9,12\}&010010011101000\\
\{1,4,11,15\}&100111010110001&\{2,12,13,15\}&111001010001101&\{5,8,9,13\}&010010011010100\\
\{1,4,12,13\}&101110001011100&\{2,12,14,15\}&010111010001011&\{5,8,9,14\}&000011011110110\\
\{1,4,12,14\}&101101110001110&\{2,13,14,15\}&011010010000111&\{5,8,9,15\}&010111011010001\\
\{1,4,12,15\}&100111110101001&\{3,4,5,6\}&101111010000100&\{5,8,10,11\}&000010010111000\\
\{1,4,13,14\}&111111101100110&\{3,4,5,7\}&001111111001010&\{5,8,10,12\}&000011010101100\\
\{1,4,13,15\}&101111100000101&\{3,4,5,8\}&001110111011110&\{5,8,10,13\}&101011010100110\\
\{1,4,14,15\}&100111100000011&\{3,4,5,9\}&001111011101110&\{5,8,10,14\}&010010010110010\\
\{1,5,6,7\}&110011110110010&\{3,4,5,10\}&001110011111010&\{5,8,10,15\}&011010011100001\\
\{1,5,6,8\}&100011011111010&\{3,4,5,11\}&001110110111000&\{5,8,11,12\}&110010011011000\\
\{1,5,6,9\}&110011101111110&\{3,4,5,12\}&001111100111010&\{5,8,11,13\}&000010110011100\\
\{1,5,6,10\}&100011000111001&\{3,4,5,13\}&001111101011100&\{5,8,11,14\}&101010110010110\\
\{1,5,6,11\}&100011111011110&\{3,4,5,14\}&001111011010010&\{5,8,11,15\}&001010010110001\\
\{1,5,6,12\}&110011111101000&\{3,4,5,15\}&001110100100001&\{5,8,12,13\}&001011011001100\\
\{1,5,6,13\}&110011110001110&\{3,4,6,7\}&001111110010000&\{5,8,12,14\}&101010010001110\\
\{1,5,6,14\}&100011111100010&\{3,4,6,8\}&001101111011000&\{5,8,12,15\}&000010111001001\\
\{1,5,6,15\}&111011100010001&\{3,4,6,9\}&001101111100100&\{5,8,13,14\}&010010110010110\\
\{1,5,7,8\}&100110110110010&\{3,4,6,10\}&001101011111100&\{5,8,13,15\}&000011011000101\\
\{1,5,7,9\}&110011111010100&\{3,4,6,11\}&101101100011000&\{5,8,14,15\}&001011011000011\\
\{1,5,7,10\}&100011101110100&\{3,4,6,12\}&001101100101000&\{5,9,10,11\}&100010001111000\\
\{1,5,7,11\}&100010101011100&\{3,4,6,13\}&001111110101100&\{5,9,10,12\}&011010001111000\\
\{1,5,7,12\}&111110100101000&\{3,4,6,14\}&001111111110110&\{5,9,10,13\}&100011001101100\\
\{1,5,7,13\}&100010111110110&\{3,4,6,15\}&101101110000001&\{5,9,10,14\}&001010101101010\\
\{1,5,7,14\}&111010111000110&\{3,4,7,8\}&001100111001100&\{5,9,10,15\}&101010001110001\\
\{1,5,7,15\}&110110111010001&\{3,4,7,9\}&001100111110000&\{5,9,11,12\}&000011001011100\\
\{1,5,8,9\}&100110011110000&\{3,4,7,10\}&001101110111110&\{5,9,11,13\}&110010101010110\\
\{1,5,8,10\}&100010011101110&\{3,4,7,11\}&001100110010110&\{5,9,11,14\}&110010001110010\\
\{1,5,8,11\}&110111010110100&\{3,4,7,12\}&001100100111100&\{5,9,11,15\}&000111001110001\\
\{1,5,8,12\}&100010010001000&\{3,4,7,13\}&001110101110100&\{5,9,12,13\}&000010011011110\\
\{1,5,8,13\}&110010010111110&\{3,4,7,14\}&001110111100010&\{5,9,12,14\}&001010001001110\\
\{1,5,8,14\}&101011010011010&\{3,4,7,15\}&001110111010001&\{5,9,12,15\}&000111101101001\\
\{1,5,8,15\}&100011111010001&\{3,4,8,9\}&101100011011000&\{5,9,13,14\}&010011101001110\\
\{1,5,9,10\}&110011011110000&\{3,4,8,10\}&001110010100000&\{5,9,13,15\}&101110011000101\\
\{1,5,9,11\}&110011001011010&\{3,4,8,11\}&101110010010000&\{5,9,14,15\}&101010111000011\\
\{1,5,9,12\}&100111001001110&\{3,4,8,12\}&111100010001000&\{5,10,11,12\}&001011000111100\\
\{1,5,9,13\}&100010001000100&\{3,4,8,13\}&011100010000100&\{5,10,11,13\}&100011000110110\\
\{1,5,9,14\}&101010101011010&\{3,4,8,14\}&101100010000010&\{5,10,11,14\}&100010100111010\\
\{1,5,9,15\}&111010111001001&\{3,4,8,15\}&001100010000001&\{5,10,11,15\}&110010010110001\\
\{1,5,10,11\}&101010001111110&\{3,4,9,10\}&001101101110010&\{5,10,12,13\}&000010000101110\\
\{1,5,10,12\}&110111001111000&\{3,4,9,11\}&101101001011010&\{5,10,12,14\}&000011000111010\\
\{1,5,10,13\}&100011100101110&\{3,4,9,12\}&011100001001000&\{5,10,12,15\}&110011001101001\\
\{1,5,10,14\}&100010000100010&\{3,4,9,13\}&111100001000100&\{5,10,13,14\}&000110010100110\\
\{1,5,10,15\}&111011111100001&\{3,4,9,14\}&001100001000010&\{5,10,13,15\}&011110010100101\\
\{1,5,11,12\}&110111000011110&\{3,4,9,15\}&101100001000001&\{5,10,14,15\}&010110010100011\\
\{1,5,11,13\}&111010100110110&\{3,4,10,11\}&111110000110000&\{5,11,12,13\}&000011100011110\\
\{1,5,11,14\}&101011001010110&\{3,4,10,12\}&101100000101000&\{5,11,12,14\}&000110010011010\\
\{1,5,11,15\}&100010000010001&\{3,4,10,13\}&001100000100100&\{5,11,12,15\}&001011010011001\\
\{1,5,12,13\}&111010011011110&\{3,4,10,14\}&111100000100010&\{5,11,13,14\}&100010000011110\\
\{1,5,12,14\}&110011010101010&\{3,4,10,15\}&001110000111001&\{5,11,13,15\}&000110010010101\\
\{1,5,12,15\}&101010110011001&\{3,4,11,12\}&001100000011000&\{5,11,14,15\}&001110010010011\\
\{1,5,13,14\}&110010110100110&\{3,4,11,13\}&101100000010100&\{5,12,13,14\}&110010001001110\\
\{1,5,13,15\}&111011011000101&\{3,4,11,14\}&011100000010010&\{5,12,13,15\}&110010010001101\\
\{1,5,14,15\}&110011011000011&\{3,4,11,15\}&111100000010001&\{5,12,14,15\}&111010010001011\\
\{1,6,7,8\}&110001110011100&\{3,4,12,13\}&111110000001100&\{5,13,14,15\}&100010010000111\\
\{1,6,7,9\}&110101101001110&\{3,4,12,14\}&101111100001010&\{6,7,8,9\}&000001111000000\\
\{1,6,7,10\}&100001111110000&\{3,4,12,15\}&001111100001001&\{6,7,8,10\}&000001111111100\\
\{1,6,7,11\}&100011110111000&\{3,4,13,14\}&001111100000110&\{6,7,8,11\}&010001110010000\\
\{1,6,7,12\}&111001110011010&\{3,4,13,15\}&001110000000101&\{6,7,8,12\}&011101110001000\\
\{1,6,7,13\}&110001111000110&\{3,4,14,15\}&111110000000011&\{6,7,8,13\}&011001111001100\\
\{1,6,7,14\}&101001111001010&\{3,5,6,7\}&001011101111110&\{6,7,8,14\}&001101110000010\\
\{1,6,7,15\}&100001101101001&\{3,5,6,8\}&001011110110010&\{6,7,8,15\}&000111110011001\\
\{1,6,8,9\}&111001011100100&\{3,5,6,9\}&101011001101010&\{6,7,9,10\}&010001101100000\\
\{1,6,8,10\}&111011010101100&\{3,5,6,10\}&001011111101000&\{6,7,9,11\}&000001101010110\\
\{1,6,8,11\}&111101010011100&\{3,5,6,11\}&101011100010100&\{6,7,9,12\}&001001101101100\\
\{1,6,8,12\}&110001011011110&\{3,5,6,12\}&101011100101000&\{6,7,9,13\}&100101101000100\\
\{1,6,8,13\}&101001110101100&\{3,5,6,13\}&001011110001110&\{6,7,9,14\}&010101101000010\\
\{1,6,8,14\}&111101011000110&\{3,5,6,14\}&001011101000010&\{6,7,9,15\}&110101101000001\\
\{1,6,8,15\}&100001010000001&\{3,5,6,15\}&101111000100001&\{6,7,10,11\}&000001100110000\\
\{1,6,9,10\}&100111001110010&\{3,5,7,8\}&001011111010100&\{6,7,10,12\}&010011100101000\\
\{1,6,9,11\}&100001001111110&\{3,5,7,9\}&001010111111100&\{6,7,10,13\}&101101100100100\\
\{1,6,9,12\}&110001101101100&\{3,5,7,10\}&001111101100000&\{6,7,10,14\}&100101100100010\\
\{1,6,9,13\}&101101011001100&\{3,5,7,11\}&011010110010000&\{6,7,10,15\}&011011100100001\\
\{1,6,9,14\}&100001001000010&\{3,5,7,12\}&101110110001000&\{6,7,11,12\}&110011100011000\\
\{1,6,9,15\}&101111111001001&\{3,5,7,13\}&001011100100100&\{6,7,11,13\}&110101100010100\\
\{1,6,10,11\}&110001100110110&\{3,5,7,14\}&001010110100110&\{6,7,11,14\}&100011100010010\\
\{1,6,10,12\}&110001000101110&\{3,5,7,15\}&101110100010001&\{6,7,11,15\}&100101100010001\\
\{1,6,10,13\}&100001000100100&\{3,5,8,9\}&001010111000000&\{6,7,12,13\}&000001100001100\\
\{1,6,10,14\}&101001011101110&\{3,5,8,10\}&001011011110000&\{6,7,12,14\}&110001100001010\\
\{1,6,10,15\}&110101010101001&\{3,5,8,11\}&001010010111110&\{6,7,12,15\}&010001100001001\\
\{1,6,11,12\}&100001000011000&\{3,5,8,12\}&011010010001000&\{6,7,13,14\}&010001100000110\\
\{1,6,11,13\}&110011000111100&\{3,5,8,13\}&111010010000100&\{6,7,13,15\}&110001100000101\\
\{1,6,11,14\}&111001101010110&\{3,5,8,14\}&001010010000010&\{6,7,14,15\}&000001100000011\\
\{1,6,11,15\}&110101110110001&\{3,5,8,15\}&101010010000001&\{6,8,9,10\}&011001011101000\\
\{1,6,12,13\}&101101000111100&\{3,5,9,10\}&101110001100000&\{6,8,9,11\}&100001011010100\\
\{1,6,12,14\}&111101000001010&\{3,5,9,11\}&011011001010000&\{6,8,9,12\}&100001011101000\\
\{1,6,12,15\}&110011010011001&\{3,5,9,12\}&111010001001000&\{6,8,9,13\}&000101011000110\\
\{1,6,13,14\}&101011101001110&\{3,5,9,13\}&011010001000100&\{6,8,9,14\}&110001011100010\\
\{1,6,13,15\}&100001110100101&\{3,5,9,14\}&101010001000010&\{6,8,9,15\}&110001011010001\\
\{1,6,14,15\}&110101000000011&\{3,5,9,15\}&001010001000001&\{6,8,10,11\}&101001010110100\\
\{1,7,8,9\}&111000111101000&\{3,5,10,11\}&001010100110000&\{6,8,10,12\}&000001010111110\\
\{1,7,8,10\}&111010110100000&\{3,5,10,12\}&001010000101000&\{6,8,10,13\}&000001011100100\\
\{1,7,8,11\}&101000110111000&\{3,5,10,13\}&101010000100100&\{6,8,10,14\}&100001010110010\\
\{1,7,8,12\}&110000111101110&\{3,5,10,14\}&011010000100010&\{6,8,10,15\}&010001011100001\\
\{1,7,8,13\}&101010111001100&\{3,5,10,15\}&111010000100001&\{6,8,11,12\}&000001110011010\\
\{1,7,8,14\}&100000110000010&\{3,5,11,12\}&101010000011000&\{6,8,11,13\}&100001110010110\\
\{1,7,8,15\}&101011110110001&\{3,5,11,13\}&001010000010100&\{6,8,11,14\}&101001011010010\\
\{1,7,9,10\}&110010101101010&\{3,5,11,14\}&111010000010010&\{6,8,11,15\}&000001010110001\\
\{1,7,9,11\}&110000101111000&\{3,5,11,15\}&011010000010001&\{6,8,12,13\}&100001010001110\\
\{1,7,9,12\}&101100101101010&\{3,5,12,13\}&101011000001100&\{6,8,12,14\}&000011011001010\\
\{1,7,9,13\}&110100111001100&\{3,5,12,14\}&001010110011010&\{6,8,12,15\}&001001111001001\\
\{1,7,9,14\}&110011101000010&\{3,5,12,15\}&101010101101001&\{6,8,13,14\}&010101010010110\\
\{1,7,9,15\}&100000101000001&\{3,5,13,14\}&101110000000110&\{6,8,13,15\}&101001111000101\\
\{1,7,10,11\}&111000110110010&\{3,5,13,15\}&011011000000101&\{6,8,14,15\}&010101011000011\\
\{1,7,10,12\}&100000100101000&\{3,5,14,15\}&101011000000011&\{6,9,10,11\}&101001001111000\\
\{1,7,10,13\}&110100100111100&\{3,6,7,8\}&101011110000010&\{6,9,10,12\}&010001001111000\\
\{1,7,10,14\}&100000110111110&\{3,6,7,9\}&001001101010000&\{6,9,10,13\}&001001001110100\\
\{1,7,10,15\}&110010110101001&\{3,6,7,10\}&101001101100000&\{6,9,10,14\}&001001011100010\\
\{1,7,11,12\}&101100110011010&\{3,6,7,11\}&001001110011100&\{6,9,10,15\}&101001011100001\\
\{1,7,11,13\}&100000100010100&\{3,6,7,12\}&001001111111010&\{6,9,11,12\}&000001011011000\\
\{1,7,11,14\}&110100101011010&\{3,6,7,13\}&101001100000110&\{6,9,11,13\}&011001011010100\\
\{1,7,11,15\}&101100101011001&\{3,6,7,14\}&001001100001010&\{6,9,11,14\}&010001011010010\\
\{1,7,12,13\}&111010110011100&\{3,6,7,15\}&001101101000001&\{6,9,11,15\}&001001011010001\\
\{1,7,12,14\}&110110100101110&\{3,6,8,9\}&001101011000000&\{6,9,12,13\}&000001001001110\\
\{1,7,12,15\}&110111100001001&\{3,6,8,10\}&001001110100000&\{6,9,12,14\}&100001101011010\\
\{1,7,13,14\}&110100110010110&\{3,6,8,11\}&001001011011110&\{6,9,12,15\}&001011001101001\\
\{1,7,13,15\}&101010110100101&\{3,6,8,12\}&101001010001000&\{6,9,13,14\}&100001101100110\\
\{1,7,14,15\}&101100100000011&\{3,6,8,13\}&001001010000100&\{6,9,13,15\}&100101011000101\\
\{1,8,9,10\}&100000011111100&\{3,6,8,14\}&111001010000010&\{6,9,14,15\}&011001111000011\\
\{1,8,9,11\}&110000011110110&\{3,6,8,15\}&011001010000001&\{6,10,11,12\}&010001100111010\\
\{1,8,9,12\}&101000011111010&\{3,6,9,10\}&101101001100110&\{6,10,11,13\}&100001100111100\\
\{1,8,9,13\}&100100011011110&\{3,6,9,11\}&001111001111000&\{6,10,11,14\}&000001001110010\\
\{1,8,9,14\}&101110011001010&\{3,6,9,12\}&001001001001000&\{6,10,11,15\}&100001001110001\\
\{1,8,9,15\}&110001111001001&\{3,6,9,13\}&101001001000100&\{6,10,12,13\}&001001000101110\\
\{1,8,10,11\}&111000011110000&\{3,6,9,14\}&011001001000010&\{6,10,12,14\}&000001101101010\\
\{1,8,10,12\}&101010110101010&\{3,6,9,15\}&111001001000001&\{6,10,12,15\}&001101010101001\\
\{1,8,10,13\}&110101010100110&\{3,6,10,11\}&101011000110000&\{6,10,13,14\}&000001110100110\\
\{1,8,10,14\}&111001110100110&\{3,6,10,12\}&111001000101000&\{6,10,13,15\}&001011010100101\\
\{1,8,10,15\}&111110010101001&\{3,6,10,13\}&011001000100100&\{6,10,14,15\}&000011010100011\\
\{1,8,11,12\}&110101010011010&\{3,6,10,14\}&101001000100010&\{6,11,12,13\}&010001000011110\\
\{1,8,11,13\}&111100010110100&\{3,6,10,15\}&001001000100001&\{6,11,12,14\}&101001100111010\\
\{1,8,11,14\}&100110010010110&\{3,6,11,12\}&011001000011000&\{6,11,12,15\}&000001101011001\\
\{1,8,11,15\}&101100010110001&\{3,6,11,13\}&001101100010100&\{6,11,13,14\}&101001000011110\\
\{1,8,12,13\}&110011011001100&\{3,6,11,14\}&001001000010010&\{6,11,13,15\}&110101010010101\\
\{1,8,12,14\}&110010110011010&\{3,6,11,15\}&101001000010001&\{6,11,14,15\}&000101010010011\\
\{1,8,12,15\}&101000011001001&\{3,6,12,13\}&001101000001100&\{6,12,13,14\}&011001010001110\\
\{1,8,13,14\}&110100010001110&\{3,6,12,14\}&001011010101010&\{6,12,13,15\}&000001010001101\\
\{1,8,13,15\}&100110010100101&\{3,6,12,15\}&011101000001001&\{6,12,14,15\}&001001010001011\\
\{1,8,14,15\}&100010110100011&\{3,6,13,14\}&011101000000110&\{6,13,14,15\}&010001010000111\\
\{1,9,10,11\}&110100001111110&\{3,6,13,15\}&111101000000101&\{7,8,9,10\}&000000111101000\\
\{1,9,10,12\}&110101001101010&\{3,6,14,15\}&001101000000011&\{7,8,9,11\}&001000111010010\\
\{1,9,10,13\}&110010011100100&\{3,7,8,9\}&101100111000000&\{7,8,9,12\}&011000111011000\\
\{1,9,10,14\}&110011001100110&\{3,7,8,10\}&001000111101110&\{7,8,9,13\}&100001111001100\\
\{1,9,10,15\}&110100101101001&\{3,7,8,11\}&101001110010000&\{7,8,9,14\}&000010111000110\\
\{1,9,11,12\}&111011001011100&\{3,7,8,12\}&001000110001000&\{7,8,9,15\}&001000111100001\\
\{1,9,11,13\}&101000101110100&\{3,7,8,13\}&101000110000100&\{7,8,10,11\}&001000110110100\\
\{1,9,11,14\}&101100001110010&\{3,7,8,14\}&011000110000010&\{7,8,10,12\}&010000110111000\\
\{1,9,11,15\}&101011001011001&\{3,7,8,15\}&111000110000001&\{7,8,10,13\}&100000111100100\\
\{1,9,12,13\}&110110001101100&\{3,7,9,10\}&011010101100000&\{7,8,10,14\}&010000111100010\\
\{1,9,12,14\}&111110001101010&\{3,7,9,11\}&001000101111000&\{7,8,10,15\}&011000110110001\\
\{1,9,12,15\}&111000001101001&\{3,7,9,12\}&101000101001000&\{7,8,11,12\}&100000111011000\\
\{1,9,13,14\}&111110001010110&\{3,7,9,13\}&001000101000100&\{7,8,11,13\}&000000111010100\\
\{1,9,13,15\}&111100111000101&\{3,7,9,14\}&111000101000010&\{7,8,11,14\}&000000110110010\\
\{1,9,14,15\}&110100111000011&\{3,7,9,15\}&011000101000001&\{7,8,11,15\}&010000111010001\\
\{1,10,11,12\}&101010100111100&\{3,7,10,11\}&101100100110000&\{7,8,12,13\}&111000110001110\\
\{1,10,11,13\}&110110000110110&\{3,7,10,12\}&011000100101000&\{7,8,12,14\}&000000110001110\\
\{1,10,11,14\}&101110000111010&\{3,7,10,13\}&101100110100110&\{7,8,12,15\}&000001110101001\\
\{1,10,11,15\}&101101101110001&\{3,7,10,14\}&001000100100010&\{7,8,13,14\}&001001111000110\\
\{1,10,12,13\}&100010110101100&\{3,7,10,15\}&101000100100001&\{7,8,13,15\}&100010111000101\\
\{1,10,12,14\}&101101010101010&\{3,7,11,12\}&111000100011000&\{7,8,14,15\}&100001111000011\\
\{1,10,12,15\}&111001110101001&\{3,7,11,13\}&011000100010100&\{7,9,10,11\}&100000101110010\\
\{1,10,13,14\}&111010000101110&\{3,7,11,14\}&101000100010010&\{7,9,10,12\}&000000101111110\\
\{1,10,13,15\}&110011010100101&\{3,7,11,15\}&001000100010001&\{7,9,10,13\}&001100101100110\\
\{1,10,14,15\}&111100110100011&\{3,7,12,13\}&001010100001100&\{7,9,10,14\}&101000111100010\\
\{1,11,12,13\}&111100000011110&\{3,7,12,14\}&101000100101110&\{7,9,10,15\}&111000101110001\\
\{1,11,12,14\}&111100100111010&\{3,7,12,15\}&101001100001001&\{7,9,11,12\}&000100101011100\\
\{1,11,12,15\}&101101010011001&\{3,7,13,14\}&011010100000110&\{7,9,11,13\}&010000101110100\\
\{1,11,13,14\}&110011010010110&\{3,7,13,15\}&001001100000101&\{7,9,11,14\}&001100101011010\\
\{1,11,13,15\}&111110010010101&\{3,7,14,15\}&001010100000011&\{7,9,11,15\}&101000111010001\\
\{1,11,14,15\}&100011010010011&\{3,8,9,10\}&101100011100100&\{7,9,12,13\}&000110101001110\\
\{1,12,13,14\}&111001001001110&\{3,8,9,11\}&001010011011000&\{7,9,12,14\}&100000101001110\\
\{1,12,13,15\}&100111010001101&\{3,8,9,12\}&101010011101000&\{7,9,12,15\}&100100111001001\\
\{1,12,14,15\}&110001010001011&\{3,8,9,13\}&011011011000110&\{7,9,13,14\}&100100111000110\\
\{1,13,14,15\}&110111010000111&\{3,8,9,14\}&001000011110110&\{7,9,13,15\}&000100111000101\\
\{2,3,4,5\}&011110000000000&\{3,8,9,15\}&001110011001001&\{7,9,14,15\}&001100111000011\\
\{2,3,4,6\}&111101001010000&\{3,8,10,11\}&011100010111000&\{7,10,11,12\}&000100100111010\\
\{2,3,4,7\}&011111100001100&\{3,8,10,12\}&001100011101000&\{7,10,11,13\}&000010100110110\\
\{2,3,4,8\}&011110111010100&\{3,8,10,13\}&001101010100110&\{7,10,11,14\}&001001100110110\\
\{2,3,4,9\}&011111101101010&\{3,8,10,14\}&101010010110010&\{7,10,11,15\}&000000101110001\\
\{2,3,4,10\}&011110110110010&\{3,8,10,15\}&001111011100001&\{7,10,12,13\}&100100101101100\\
\{2,3,4,11\}&011111110011010&\{3,8,11,12\}&001001010111000&\{7,10,12,14\}&100001110101010\\
\{2,3,4,12\}&011111011011000&\{3,8,11,13\}&001011010010110&\{7,10,12,15\}&010100110101001\\
\{2,3,4,13\}&011111101010110&\{3,8,11,14\}&001101010011010&\{7,10,13,14\}&010000100101110\\
\{2,3,4,14\}&011110001011010&\{3,8,11,15\}&001101110110001&\{7,10,13,15\}&001100110100101\\
\{2,3,4,15\}&011100111001001&\{3,8,12,13\}&101000010011100&\{7,10,14,15\}&000100110100011\\
\{2,3,5,6\}&011111100110000&\{3,8,12,14\}&001000011001010&\{7,11,12,13\}&001000100011110\\
\{2,3,5,7\}&011011100101110&\{3,8,12,15\}&001010110101001&\{7,11,12,14\}&110000100011110\\
\{2,3,5,8\}&111011010010000&\{3,8,13,14\}&011000010100110&\{7,11,12,15\}&100001110011001\\
\{2,3,5,9\}&011011111011110&\{3,8,13,15\}&001000011000101&\{7,11,13,14\}&100100100110110\\
\{2,3,5,10\}&011111011100100&\{3,8,14,15\}&011110011000011&\{7,11,13,15\}&010100110010101\\
\{2,3,5,11\}&011011011111010&\{3,9,10,11\}&001100001111110&\{7,11,14,15\}&100100110010011\\
\{2,3,5,12\}&011011110111000&\{3,9,10,12\}&001101001101010&\{7,12,13,14\}&011000101001110\\
\{2,3,5,13\}&011111111111100&\{3,9,10,13\}&101000001101100&\{7,12,13,15\}&100000110001101\\
\{2,3,5,14\}&011110101111110&\{3,9,10,14\}&001011001100110&\{7,12,14,15\}&010000110001011\\
\{2,3,5,15\}&011111110101001&\{3,9,10,15\}&001100101101001&\{7,13,14,15\}&110000110000111\\
\{2,3,6,7\}&011001100000000&\{3,9,11,12\}&001011001011010&\{8,9,10,11\}&000000011110000\\
\{2,3,6,8\}&011101011110110&\{3,9,11,13\}&001010101010110&\{8,9,10,12\}&110100011101000\\
\{2,3,6,9\}&011011101110100&\{3,9,11,14\}&101100101010110&\{8,9,10,13\}&010100011100100\\
\{2,3,6,10\}&111101010100000&\{3,9,11,15\}&001011101110001&\{8,9,10,14\}&000010011100010\\
\{2,3,6,11\}&011001110010110&\{3,9,12,13\}&001000001011100&\{8,9,10,15\}&100010011100001\\
\{2,3,6,12\}&011111010111110&\{3,9,12,14\}&101100001001110&\{8,9,11,12\}&010100011011000\\
\{2,3,6,13\}&011101111101110&\{3,9,12,15\}&011011011001001&\{8,9,11,13\}&110100011010100\\
\{2,3,6,14\}&011111110100110&\{3,9,13,14\}&001101101001110&\{8,9,11,14\}&100010011010010\\
\{2,3,6,15\}&011111001000001&\{3,9,13,15\}&011101011000101&\{8,9,11,15\}&000010011010001\\
\{2,3,7,8\}&111100110010000&\{3,9,14,15\}&111000011000011&\{8,9,12,13\}&000000011001100\\
\{2,3,7,9\}&011001111110000&\{3,10,11,12\}&111000000111100&\{8,9,12,14\}&110000011001010\\
\{2,3,7,10\}&011100110100000&\{3,10,11,13\}&001111010110100&\{8,9,12,15\}&010000011001001\\
\{2,3,7,11\}&111001100110000&\{3,10,11,14\}&001010001110010&\{8,9,13,14\}&101000011000110\\
\{2,3,7,12\}&011000110111110&\{3,10,11,15\}&001100001110001&\{8,9,13,15\}&110000011000101\\
\{2,3,7,13\}&011100110011100&\{3,10,12,13\}&101110010101100&\{8,9,14,15\}&000000011000011\\
\{2,3,7,14\}&011010111110110&\{3,10,12,14\}&001000000111010&\{8,10,11,12\}&100100010111000\\
\{2,3,7,15\}&011011111010001&\{3,10,12,15\}&011000010101001&\{8,10,11,13\}&100010010110100\\
\{2,3,8,9\}&011000011000000&\{3,10,13,14\}&101000000110110&\{8,10,11,14\}&001100010110010\\
\{2,3,8,10\}&011011010100000&\{3,10,13,15\}&101101010100101&\{8,10,11,15\}&100000110110001\\
\{2,3,8,11\}&011101110110100&\{3,10,14,15\}&001000010100011&\{8,10,12,13\}&001000010101100\\
\{2,3,8,12\}&011100111111010&\{3,11,12,13\}&001111000011110&\{8,10,12,14\}&000000010101010\\
\{2,3,8,13\}&011100011011110&\{3,11,12,14\}&111000001011010&\{8,10,12,15\}&100000010101001\\
\{2,3,8,14\}&011010011101110&\{3,11,12,15\}&001010101011001&\{8,10,13,14\}&100000010100110\\
\{2,3,8,15\}&111000010011001&\{3,11,13,14\}&101101010010110&\{8,10,13,15\}&000000010100101\\
\{2,3,9,10\}&011110001100110&\{3,11,13,15\}&011000010010101&\{8,10,14,15\}&110000010100011\\
\{2,3,9,11\}&011101111010010&\{3,11,14,15\}&101000010010011&\{8,11,12,13\}&010000010011100\\
\{2,3,9,12\}&011111001001110&\{3,12,13,14\}&001110100101110&\{8,11,12,14\}&100000010011010\\
\{2,3,9,13\}&011110011001100&\{3,12,13,15\}&001010010001101&\{8,11,12,15\}&000000010011001\\
\{2,3,9,14\}&011001101011010&\{3,12,14,15\}&011100010001011&\{8,11,13,14\}&000000010010110\\
\{2,3,9,15\}&011111101011001&\{3,13,14,15\}&101001010000111&\{8,11,13,15\}&100000010010101\\
\{2,3,10,11\}&011000000110000&\{4,5,6,7\}&000111100000000&\{8,11,14,15\}&010000010010011\\
\{2,3,10,12\}&011001100111100&\{4,5,6,8\}&000111110010110&\{8,12,13,14\}&001100010001110\\
\{2,3,10,13\}&011100100110110&\{4,5,6,9\}&010111001001000&\{8,12,13,15\}&010100010001101\\
\{2,3,10,14\}&011101100100010&\{4,5,6,10\}&000111101100110&\{8,12,14,15\}&000010010001011\\
\{2,3,10,15\}&011001101101001&\{4,5,6,11\}&000111100111100&\{8,13,14,15\}&000100010000111\\
\{2,3,11,12\}&011010101011100&\{4,5,6,12\}&000111110101010&\{9,10,11,12\}&000100001111000\\
\{2,3,11,13\}&011101100011110&\{4,5,6,13\}&000111111001100&\{9,10,11,13\}&000010001110100\\
\{2,3,11,14\}&011010100111010&\{4,5,6,14\}&000111001111110&\{9,10,11,14\}&010100001110010\\
\{2,3,11,15\}&011110010011001&\{4,5,6,15\}&010111000100001&\{9,10,11,15\}&110100001110001\\
\{2,3,12,13\}&011000000001100&\{4,5,7,8\}&000110111011000&\{9,10,12,13\}&010000001101100\\
\{2,3,12,14\}&011001110101010&\{4,5,7,9\}&000110111100100&\{9,10,12,14\}&100000001101010\\
\{2,3,12,15\}&011010100001001&\{4,5,7,10\}&000111111110000&\{9,10,12,15\}&000000001101001\\
\{2,3,13,14\}&011001101100110&\{4,5,7,11\}&100111100110000&\{9,10,13,14\}&000000001100110\\
\{2,3,13,15\}&111010100000101&\{4,5,7,12\}&110110101001000&\{9,10,13,15\}&100000001100101\\
\{2,3,14,15\}&011000000000011&\{4,5,7,13\}&001110110000100&\{9,10,14,15\}&010000001100011\\
\{2,4,5,6\}&011111111000000&\{4,5,7,14\}&000110110111110&\{9,11,12,13\}&110000001011100\\
\{2,4,5,7\}&010110111010010&\{4,5,7,15\}&100110101110001&\{9,11,12,14\}&000000001011010\\
\{2,4,5,8\}&110111010001000&\{4,5,8,9\}&000110011000000&\{9,11,12,15\}&100000001011001\\
\{2,4,5,9\}&010110111101110&\{4,5,8,10\}&110110010100000&\{9,11,13,14\}&100000001010110\\
\{2,4,5,10\}&010111101101100&\{4,5,8,11\}&010110010010000&\{9,11,13,15\}&000000001010101\\
\{2,4,5,11\}&010111111111010&\{4,5,8,12\}&001111010001000&\{9,11,14,15\}&110000001010011\\
\{2,4,5,12\}&010111110011100&\{4,5,8,13\}&000110011111100&\{9,12,13,14\}&010100001001110\\
\{2,4,5,13\}&010110011110110&\{4,5,8,14\}&011111010000010&\{9,12,13,15\}&110100001001101\\
\{2,4,5,14\}&010111100110110&\{4,5,8,15\}&100110110000001&\{9,12,14,15\}&000100001001011\\
\{2,4,5,15\}&110111000010001&\{4,5,9,10\}&010110001100000&\{9,13,14,15\}&100100001000111\\
\{2,4,6,7\}&010101111010100&\{4,5,9,11\}&001110001010000&\{10,11,12,13\}&000000000111100\\
\{2,4,6,8\}&010111111000110&\{4,5,9,12\}&000111011101000&\{10,11,12,14\}&110000000111010\\
\{2,4,6,9\}&010101111101000&\{4,5,9,13\}&001110001101100&\{10,11,12,15\}&000101000111001\\
\{2,4,6,10\}&010111110100000&\{4,5,9,14\}&100110101000010&\{10,11,13,14\}&010000000110110\\
\{2,4,6,11\}&010101011110000&\{4,5,9,15\}&100110001101001&\{10,11,13,15\}&110000000110101\\
\{2,4,6,12\}&010111100001010&\{4,5,10,11\}&000110000110000&\{10,11,14,15\}&000000000110011\\
\{2,4,6,13\}&010101110001110&\{4,5,10,12\}&011111000101000&\{10,12,13,14\}&100100000101110\\
\{2,4,6,14\}&010101001100110&\{4,5,10,13\}&000111000100100&\{10,12,13,15\}&000100000101101\\
\{2,4,6,15\}&011101111100001&\{4,5,10,14\}&000111010110010&\{10,12,14,15\}&110100000101011\\
\{2,4,7,8\}&010110110110100&\{4,5,10,15\}&000110010101001&\{10,13,14,15\}&010100000100111\\
\{2,4,7,9\}&011100101010000&\{4,5,11,12\}&000111000011000&\{11,12,13,14\}&011010000011110\\
\{2,4,7,10\}&111100101100000&\{4,5,11,13\}&100111000010100&\{11,12,13,15\}&100100000011101\\
\{2,4,7,11\}&011101101111000&\{4,5,11,14\}&101111000010010&\{11,12,14,15\}&010100000011011\\
\{2,4,7,12\}&011110111101000&\{4,5,11,15\}&001111000010001&\{11,13,14,15\}&110100000010111\\
\{2,4,7,13\}&010100111111100&\{4,5,12,13\}&000110000001100&\{12,13,14,15\}&000000000001111\\
\hline
\end{longtable}
\end{document}